\documentclass[final,5p,times]{elsarticle}

\usepackage{amssymb}   
\usepackage{amsthm}    
\usepackage{amsmath}
\usepackage{xcolor}
\usepackage{enumerate}
\usepackage{graphicx}
\usepackage[english]{babel}
\usepackage{hyphenat}
\usepackage{kantlipsum}
\biboptions{authoryear}
\usepackage[shortcuts]{extdash} 

\allowdisplaybreaks

\def\K{\mathcal{K}}
\def\Kinf{\K_{\infty}}
\def\KL{\mathcal{KL}}
\def\R{\mathbb{R}}
\def\N{\mathbb{N}}
\def\X{\mathcal{X}}

\def\U{\mathcal{U}}

\def\Ki{\K_\infty}
\def\T{\mathcal{T}}
\def\mer{\hfill $\circ$}

\def\qed{$\hfill\blacksquare$}

\newtheorem{theorem}{Theorem}[section]
\newtheorem{definition}[theorem]{Definition}
\newtheorem{remark}[theorem]{Remark}
\newtheorem{lemma}[theorem]{Lemma}

\newtheorem{assumption}[theorem]{Assumption}

\newtheorem{prop}[theorem]{Proposition}
\theoremstyle{remark}
\newtheorem{claim}{Claim}

\journal{Automatica}

\begin{document}

\begin{frontmatter}

\title{Semiglobal exponential input-to-state stability of sampled-data systems based on approximate discrete-time models}

\author{Alexis J. Vallarella}
\author{Paula Cardone}
\author{Hernan Haimovich}

\address{Centro Internacional Franco-Argentino de Ciencias de la Informaci\'on y de Sistemas (CIFASIS),
CONICET-UNR, Ocampo y Esmeralda,
2000 Rosario, Argentina. {\texttt{\{vallarella,cardone,haimovich\}@cifasis-conicet.gov.ar}}}

\begin{abstract}
Exact discrete-time models of nonlinear systems are difficult or impossible to obtain, and hence approximate models may be employed for control design. Most existing results provide conditions under which the stability of the approximate model in closed-loop carries over to the stability of the (unknown) exact model but only in a practical sense, i.e. the trajectories of the closed-loop system are ensured to converge to a bounded region whose size can be made as small as desired by limiting the maximum sampling period. In addition, some very stringent conditions exist for the exact model to exhibit exactly the same type of asymptotic stability as the approximate model. 
In this context, our main contribution consists in providing less stringent conditions by considering semiglobal exponential input-to-state  stability (SE-ISS), where the inputs can successfully represent state-measurement and actuation errors. These conditions are based on establishing SE-ISS for an adequate approximate model and are applicable both under uniform and nonuniform sampling. 
As a second contribution, we show that explicit Runge-Kutta models satisfy our conditions and can hence be employed. 
An example of control design for stabilization based on approximate discrete-time models is also given.
\end{abstract}

\end{frontmatter}

\section{Introduction}
\label{sec:introduction}

Modern digital control applications involve measuring the available signals
of the continuous-time plant via a sampling mechanism
and then applying the computed control action
via zero-order hold (ZOH).
One of the existing approaches for control design
consists in designing a discrete-time control law
based 
on a discrete\-/time model of the plant.
For nonlinear systems
the exact discrete\-/time model,
i.e. the model that exactly matches 
the state of the continuous-time system
at sampling instants, 
may be difficult (or impossible) to derive 
due to the complexity (or non-existence)
of the closed form solutions of the equations that describe 
the plant dynamics.
Thus, the usual approach is to design 
the control law 
based
on a (sufficiently good) \textit{approximate} discrete-time model.

In this context, several results have been derived
in order to establish different kinds of stability properties
of the exact model or generate
adequate approximate models \citep{NesicSCL99,Nesic2002,NesicTAC04,karafyllis2009global,nevsic2009stability,monaco2007advanced,yuz2014sampled,van2012discrete,ZENG201773}.
These results usually provide conditions that
ensure certain kind of stability property
for the approximate 
closed-loop model and,
if 
the exact and approximate models satisfy some consistency property,
also ensure that the stability property or a practical version of it is also
fulfilled by the exact model for sufficiently small sampling periods.
These stability properties contemplate a wide range of situations
with respect to the uniformity of the sampling (periodic or aperiodic), the nature of the convergence to an equilibrium point (asymptotic or practical), the consideration of disturbances (input-to-state stability properties) and the nature of the maximum allowable sampling period (semiglobal or global).
Several works
address problems such as the presence of time delays \citep{difer3,difer1,difer2},
observer design \citep{ARCAK20041931,6112660,Beikzadeh20161}, and control schemes involving dual-rate
\citep{liu2008input,USTUNTURK20121796} or multirate \citep{Beikzadeh20151939,POLUSHIN20041035} sampling.
Other approaches for sampled-data stabilization only contemplate emulated controllers \citep{NESIC_A} or require Lyapunov-like assumptions on the continuous\-/time plant \citep{NESIC_A,ABDEL_B}. A recent publication 
\citep{Wei_LIN2020} 
shows that global asymptotic and
local exponential stabilizability of the continuous-time plant by state feedback imply semiglobal asymptotic stabilizability by digital state feedback; these results hold under uniform sampling.
In \citet{VALLARELLA201860}, we derived necessary and sufficient conditions for (i) semiglobal asymptotic stability, robustly with respect to bounded disturbances, and (ii) semiglobal ISS, where the (disturbance) input may successfully represent
state-measurement or actuation errors, both for discrete-time models of nonuniformly sampled (i.e. periodic or aperiodic) nonlinear systems. These properties are \emph{semiglobal} only in the sampling period, meaning that a bound on the state exists such that for every bound on the initial condition (and input), a maximum sampling period exists for which the state bound holds. 
In \citet{IEEETAC18}, 
we have shown that if a consistency property (MSEC) 
holds between the approximate and exact closed-loop models
then the smaller the maximum admissible sampling period is, the lower the error between their solutions over a fixed time period becomes.
Moreover, if the control law
renders the approximate model semiglobal practical ISS under nonuniform sampling (SP-ISS-VSR), then the same controller
ensures SP-ISS-VSR of the exact 
closed\-/loop model.
In all of these existing results, stability of the exact model is either semiglobal and practical
or global and asymptotic. The conditions for ensuring global stability are stringent.
To the best of the authors' knowledge, results based on approximate discrete-time models, ensuring
semiglobal and asymptotic stability under conditions that may hence be much weaker than 
those required for a global result, without requiring Lyapunov-like assumptions on the continuous-time plant, and admitting nonuniform sampling and control laws not necessarily based on emulation (which may use the knowledge of the current sampling period to compute the control action), have not been previously derived.

The main purpose of this paper is thus to provide results that ensure semiglobal asymptotic 
ISS under nonuniform sampling and in the presence of disturbances 
that successfully cover the case of state-measurement and actuation errors.
Specifically, we give sufficient conditions
for semiglobal exponential ISS
under nonuniform sampling (SE-ISS-VSR) of the exact 
closed-loop  
model, based on the fact that the same property holds for an approximate model.
To do that, 
we 
introduce 
two novel consistency properties that take disturbances into account: Robust Equilibrium\-/Preserving Consistency (REPC) and Multistep Consistency (REPMC).
REPC bounds the mismatch after only one sampling period and REPMC bounds the mismatch between the models' trajectories over finite time intervals, irrespective of how many sampling periods fall within the interval.
As a second contribution, we show
that any explicit 
and consistent Runge-Kutta model
is REPC with the exact model under very mild conditions not
requiring high-order differentiability of the function
that defines the continuous-time plant.

The organization of this paper is as follows.
In Section \ref{sec:preliminaries} we present
a brief summary of the notation employed, 
we state the problem and the required definitions and
properties. Our main results are given in Section \ref{sec:mainresults2}. 
An illustrative
example of stabilization of a plant via discrete-time design is provided in Section \ref{sec:example}. Concluding remarks are
presented in Section \ref{sec:conclusions}. 
The Appendix contains the proofs of the presented results and some of the intermediate technical points.

\section{Preliminaries}
\label{sec:preliminaries}

\subsection{Notation}
\label{sec:notation}
$\R$, $\R_{\ge 0}$, $\N$ 
and $\N_0$ denote the sets of real, nonnegative real, natural and nonnegative integer numbers, respectively. 
We write $\alpha \in \K$ if $\alpha : \R_{\ge 0} \to \R_{\ge 0}$ is strictly increasing, 
continuous and $\alpha(0)=0$. We write $\alpha \in \Kinf$ if $\alpha\in\K$ and $\alpha$ is unbounded. 
We write $\beta \in \KL$ if $\beta : \R_{\ge 0} \times \R_{\ge 0} \to \R_{\ge 0}$, $\beta(\cdot,t)\in \K$ for all $t\ge 0$,
and $\beta(s,\cdot)$ is strictly decreasing asymptotically to $0$ for every $s$.
We denote the Euclidean norm of a vector $x \in \R^n$ by $|x|$. 
We denote an infinite sequence as $\{T_i\}:=\{T_i\}_{i=0}^{\infty}$. 
For any sequences $\{T_i\} \subset \R_{\ge 0}$ and $\{e_i\} \subset \R^m$, and any $\gamma\in\K$,
we take the following conventions: $\sum_{i=0}^{-1} T_i = 0$ and $\gamma(\sup_{0\le i\le -1}|e_i|) = 0$. 
Given a real number $T>0$ we denote by 
$\Phi(T):=\{ \{T_i\} : \{T_i\} \text{ is such that }  T_i \in (0,T) \text{ for all } i\in \N_0 \}$
the set of all sequences of real numbers in the open interval $(0,T)$. 
For a given sequence we denote the norm $\|\{x_i\}\|:= \sup_{i\geq0} |x_i|$.

\subsection{Discrete-time models}
\label{sec:problem_statement}
We consider discrete-time models for sampled continuous-time nonlinear systems of the form
\begin{equation}
\label{eq:cs}
 \dot{x}=f(x,u),\quad  x(0)=x_0, 
\end{equation}
under zero-order hold,
where $x(t) \in \R^n$, $u(t) \in \R^m$ are the state and control vectors respectively.
We consider that the sampling instants 
$t_k$, $k\in \N_0$, satisfy $t_0 = 0$ and $t_{k+1} = t_k + T_k$, where $\{T_k\}_{k=0}^{\infty}$ is the sequence of corresponding sampling periods.
We consider that sampling periods may vary;
we refer to this scheme as Varying Sampling Rate (VSR).
We also assume that the next sampling instant $t_{k+1}$, and hence
the current sampling period $T_k$, is known at the current sampling 
instant $t_k$. 
This situation is typical of schemes
where the controller sets the next sampling instant
according to a specific control strategy
as in self-triggered control \citep{Anta2010}.
Due to zero-order hold, the continuous-time control signal is piecewise constant
such that $u(t) = u(t_k) =: u_k$ for all $t\in [t_k,t_{k+1})$.
The class of discrete-time systems that arise when modelling
\eqref{eq:cs} under this scheme is thus of the form
\begin{equation}
  \label{eq:fut}
  x_{k+1} = F^{\diamondsuit}(x_k,u_k,T_k),
\end{equation}
meaning that the state at the next sampling instant depends on the current state and input values, as well as on the current sampling period. We will set $\diamondsuit=e$ to symbolize that 
the discrete-time model is exact (i.e. its state
coincides with that of the continuous-time plant state at sampling instants). We will set
$\diamondsuit=a,b,c$, etc., for other in principle arbitrary discrete-time models, and $\diamondsuit = Euler$ or $\diamondsuit = RK$ for the Euler or a Runge-Kutta model, respectively. Using our notation and the definition of the Euler model, then $F^{Euler}(x,u,T):=x+T f(x,u)$.

Given that the current sampling period $T_k$ 
is known or 
determined at the current sampling instant $t_k$, 
the current control action $u_k$ may depend not only on the current state sample $x_k$ but also on $T_k$.
If
state-measurement or actuation errors exist
we will denote them by $e_k \in \R^q$,
where 
the dimension $q$ depends on the type of error 
(i.e., $q=n$ for state-measurement additive error or $q=m$ for actuation additive error).
In this case, the true control action applied will also be affected by such errors
\begin{equation}
  \label{eq:UxeT}
  u_k = U(x_k,e_k,T_k).
\end{equation}
This scheme also covers the case of static output
feedback. For a system's output $y=h(x) \in \R^p$ and
a control law $W(y,e,T)$ we can simply define the function $U(x,e,T):=W(h(x),e,T)$ 
to obtain \eqref{eq:UxeT}.
Under (\ref{eq:UxeT}), the closed-loop model given by the pair $(U,F^{\diamondsuit})$ becomes
\begin{equation}
  \label{eq:system1}
  x_{k+1} =  F^{\diamondsuit}(x_k, U(x_k,e_k,T_k),T_k) =: \bar F^{\diamondsuit}(x_k,e_k,T_k)
\end{equation}
which is once again of the form \eqref{eq:fut}.
For the sake of notation,
we may refer to the discrete-time model  \eqref{eq:system1} simply as $\bar F^{\diamondsuit}$.

\subsection{Definitions and stability properties}

We consider that the continuous-time 
model of the plant \eqref{eq:cs} and the control law \eqref{eq:UxeT} 
fulfill
the following 
assumptions.

\begin{assumption}
\label{def:UniLipschitz}
The function  $f:\R^n \times \R^m \rightarrow \R^n$
is locally Lipschitz in $x$ uniformly in $u$, i.e. 
for every compact sets $\mathcal{X}\subset \R^n$ 
$\mathcal{U}\subset \R^m$
there exists
 $L=L(\X,\U)>0$ such that for all $x,y \in \X$ and $u\in \mathcal{U}$
we have $\left|f(x,u)-f(y,u)\right| \leq L|x-y|$.
\end{assumption}

\begin{assumption}
\label{def:bounded}
The function $f$ is locally bounded, i.e.
for every $M,C_u\geq 0$ there exists
$C_f=C_f(M,C_u)>0$, 
with $C_f(\cdot,\cdot)$ nondecreasing in each variable,
such that 
$|f(x,u)| \leq C_f$ 
for every $|x|\leq M$ and $|u| \leq C_u$.
\end{assumption}

\begin{assumption}
\label{def:controllaw}
The control law $U(x,e,T)$ is
\textit{small-time locally uniformly bounded}, i.e.
for every $M,E\geq 0$ there exist $T^{u}=T^{u}(M,E)>0$ and $C_u=C_u(M,E)>0$, with $T^u(\cdot,\cdot)$ nonincreasing in each variable and $C_u(\cdot,\cdot)$ nondecreasing in each variable,
such that $|U(x,e,T)| \leq C_u$ for all $|x|\leq M$,
$|e|\leq E$,
and $T\in (0,T^{u})$. 
\end{assumption}

\begin{remark}
Assumption \ref{def:controllaw} ensures that there exists
a maximum sampling period such that the control law remains bounded
for all states and disturbances whose norms are bounded by 
$M,E\ge 0$, respectively. For example, the control law $U(x,e,T)=-x/(1-T|x|)$ is
small-time locally uniformly bounded with $T^u=1/(2M)$ 
but $U(x,e,T)=-x/(T(1-|x|))$ is not, as it grows unbounded for every $T>0$ when $|x|\rightarrow 1$.
\end{remark}

The following stability definitions are used
throughout the paper.

\begin{definition} 
\label{def:SE-ISS-VSR}
  The system (\ref{eq:system1}) is said to be
  \begin{enumerate} [i)]
    \item 
  \label{def:ISS}
   \textit{Semiglobally ISS-VSR} (S-ISS-VSR) if there exist 
 $\beta \in \KL$ and $\gamma \in \K_\infty$
    such that for all $M,E\geq0$ there exists
    $T^{\star}=T^{\star}(M,E)>0$ such that 
   for all $k\in \N_0$, $\{T_i\} \in \Phi(T^\star)$, $|x_0|\leq M$ and $\|\{e_i\}\|\leq E$ 
    the solutions of \eqref{eq:system1} satisfy
      \begin{equation}
      \label{eq:S-ISS-VSR}
      |x_k|\leq \beta\left(|x_0|,\sum_{i=0}^{k-1}T_i\right)+\gamma\left(\sup_{0\le i \le k-1}|e_i| \right ).
    \end{equation}
    
\item   \label{def:SEISS}
   \textit{Semiglobally Exponentially ISS-VSR} (SE-ISS-VSR)
   if it is S-ISS-VSR and additionally $\beta \in \KL$
   can be chosen as $\beta(r,t):=K r \exp(-\lambda t)$
   with   $K\geq 1$ and $\lambda>0$. 
   
\end{enumerate}
\end{definition}

The S-ISS-VSR property 
was introduced in \citet{VALLARELLA201860}.
Since the maximum admisible sampling period $T^\star$ depends
on the bound of the initial condition and error input
it constitutes a natural semiglobal version 
of the ISS property for discrete-time models 
under nonuniform sampling. The fact that it holds for
all posible sequences of sampling periods that are bounded
by $T^\star$ makes it useful in linking the stability of the sampled-data system with that of its discrete-time model.

\section{Main results}
\label{sec:mainresults2}

\subsection{SE-ISS-VSR via approximate discrete-time models}
\label{subsec:SE-ISS}

In this section, we give novel sufficient conditions for SE-ISS-VSR of the exact discrete-time model based on an approximate model. For this, we introduce two novel consistency properties:
Robust Equilibrium\-/Preserving Consistency (REPC),
which is a one-step property, and Robust Equilibrium\-/Preserving Multistep Consistency (REPMC).
Specifically, we will prove that
if the approximate closed-loop model is SE-ISS-VSR and if the exact and approximate models are REPMC, 
then the exact closed-loop model is also SE-ISS-VSR. 

\begin{definition}
\label{def:HOA}
The discrete-time model $\bar F^a$ 
is said to be 
Robustly Equilibrium\-/Preserving Consistent (REPC)
with  $\bar F^b$ if
there exists $\phi \in \K_\infty$ such that 
for each $M,E\geq 0$
there exist constants 
$K:=K(M,E)>0$, $T^*:=T^*(M,E)>0$ and a function $\rho \in \K_\infty$
such that 
\begin{align}
\label{eq:kraftwerk}
&\left|\bar F^a (x^a,e,T)- \bar F^b(x^b,e,T) \right|  \notag \\
&\quad \leq (1+KT)\left|x^a-x^b\right|+ T \rho(T)  \left(\max\{|x^a|,|x^b|\}+\phi(|e|)\right)
\end{align}
for all $|x^a|,|x^b| \leq M$, $|e|\leq E$ and $T\in (0,T^*)$.
The pair $(\bar F^a, \bar F^b)$ is said to be REPC if $\bar F^a$ is REPC with $\bar F^b$.
\end{definition}

The REPC
condition is robust in the sense that it admits the presence of 
discrete-time bounded disturbances and ensures that their effect on the mismatch between
models in one step is bounded by a quantity that can be reduced by decreasing the sampling period. 
REPC additionally requires $\bar F^a(0,0,T) = \bar F^b(0,0,T)$, and thus forces the mismatch between models to approach 0 as the equilibrium is approached in the absence of disturbances. The latter feature is key in allowing any type of asymptotic stability to be mirrored from one model to the other.

It is evident that REPC 
is symmetric (if $\bar F^a$ is REPC with $\bar F^b$, then $\bar F^b$ is REPC with $\bar F^a$). 
REPC is also 
transitive, as stated in Proposition~\ref{cor:1} and proven in~\ref{app:proof:cor:1}. 
\begin{prop}
\label{cor:1}
Suppose that the pairs $(\bar F^a, \bar F^b)$
and $(\bar F^b, \bar F^c)$ are REPC.
Then $(\bar F^a, \bar F^c)$ is REPC.
\end{prop}
We next introduce the REPMC property,
which extends the 
linear gain multistep upper consistency 
property
in \cite[Definition~5]{nevsic2009stability} 
by the facts that: (i) it
is a perturbation-admitting condition, and (ii) it is
semiglobal with respect to the magnitude of the initial condition
and disturbance input.
%
\begin{definition}
\label{eq:default}
The discrete-time model $\bar F^a$
is said to be 
Robustly Equilibrium\-/Preserving Multistep Consistent
(REPMC)
with  $\bar F^b$ if 
there exists $\phi \in \K_\infty$
such that 
for each $M,E\geq0$
and $\T,\eta>0$
there exist a constant $T^*=T^*(M,E,\T,\eta)>0$ 
and a function 
$\alpha:\R_{\geq0} \times \R_{\geq0}  \rightarrow \R_{\geq0} \cup \{\infty\} $
with $\alpha(\cdot,T)$ non-decreasing for all $T\in[0,T^*)$
such that 
\begin{equation}
|x^a-x^b|\leq\delta \Rightarrow|\bar F^a(x^a,e,T)-\bar F^b(x^b,e,T)|
\leq \alpha(\delta,T)
\label{eq:orwell}
\end{equation}
for all $|x^a|,|x^b| \leq M$, $|e| \leq E$ and $T\in(0,T^*)$, 
and $\sum_{i=0}^{k-1} T_i \leq \T$ implies
\vspace{-0.4cm}
\begin{equation}
\alpha^k(0,\{T_i\}):=
\overbrace{\alpha(\cdots\alpha(\alpha}^{k}(0,T_0),T_1)\cdots,T_{k-1}) \leq \eta M +  \phi(E).
\label{eq:trap}
\end{equation}
\end{definition}
In \citet[Definition~2.6]{IEEETAC18}, we introduced a perturbation-admitting consistency property
called MSEC. The main difference
between MSEC and REPMC is that the latter requires the difference between model solutions to become smaller as the equilibrium is approached, and forces such a difference to be 0 at the equilibrium (case $M=E=0$). At the same time, REPMC does not require the effect of the 
disturbances on the difference between solutions to decrease as the sampling period is decreased.
Lemma~\ref{lem:menem} makes these facts more explicit; its proof is given in \ref{app:proof:menem}.
%
\begin{lemma}
\label{lem:menem}
Suppose that 
$\bar F^a$
is REPMC with 
$\bar F^b$
as per Definition~\ref{eq:default}
with function $\phi \in \K_\infty$.
Let $x^a(k,\xi,\{e_i\},\{T_i\})$ and $x^b(k,\xi,\{e_i\},\{T_i\})$
be the solutions
with
initial condition
$\xi \in \R^n$,
input sequence $\{e_i\}$
and sampling period sequence $\{T_i\}$ 
for the models $\bar F^a$ and $\bar F^b$,
respectively.
Then for each $M_a,E\geq0$ and
$\T,\eta > 0$, there exists
$T^{L}=T^L(M_a,E,\T,\eta)>0$ such that, if $\xi \in \R^n$ satisfies
\begin{equation}
|x^a(k,\xi,\{e_i\},\{T_i\})| \leq M_a
\end{equation}
for all $\{T_i\} \in \Phi(T^L)$,
$\|\{e_i\}\|\leq E$ and 
$k\in \N_0$ for which $\sum_{i=0}^{k-1} T_i\in [0,\T]$,
then
\begin{equation*}
|x^b(k,\xi,\{e_i\},\{T_i\})-x^a(k,\xi,\{e_i\},\{T_i\})| 
\leq \eta |\xi|  + \phi\left(\sup_{0\leq i \leq k-1} |e_i|\right).
\end{equation*}
for all $\{T_i\} \in \Phi(T^L)$,
$\|\{e_i\}\|\leq E$ and 
$k\in \N_0$ for which $\sum_{i=0}^{k-1} T_i\in [0,\T]$.
\end{lemma}

Our main result is the following.

\begin{theorem}
\label{theorem:SEISS}
 Consider that
 \begin{enumerate}[i)]
 \item 
 $(\bar F^a,\bar F^b)$ is REPMC 
 with $\phi\in \K_\infty$.
\label{item:hoa123333}
 \item $x^a_{k+1}=\bar F^a(x^a_k,e_k,T_k)$ is SE-ISS-VSR with $K_a\geq 1$, $\lambda_a>0$ and $\gamma_a \in \K_\infty$. \label{item:SES12}
 \end{enumerate}
Then $x^b_{k+1}=\bar F^b(x^b_k,e_k,T_k)$  is SE-ISS-VSR with $K_b\geq 1$, $\lambda_b>0$ and $\gamma_b \in \K_\infty$ given by
\begin{align*}
    &K_b:=\frac{K_a+\eta}{\delta} \medspace \text{, } \medspace \lambda_b:=- \frac{\ln(\delta)}{\frac{1}{\lambda_a} \ln(\frac{K_a}{\delta-\eta})+1}
    \medspace \text{, }  \\
    &\gamma_b:=\left(\frac{K_a+\eta}{1-\delta}+1 \right)(\gamma_a+\phi),
\end{align*}
with $0<\eta< \delta < 1$  
quantities that can be chosen arbitrarily.
\end{theorem}

The proof of Theorem~\ref{theorem:SEISS} is given in \ref{app:proof:SEISS}.
Theorem~\ref{theorem:SEISS} provides a sufficient condition, namely the REPMC property, for the SE-ISS-VSR of a (closed-loop) model $\bar F^a$ to carry over to another model $\bar F^b$ (and viceversa). Therefore, to establish SE-ISS-VSR of the exact model it suffices to ensure that some approximate model is SE-ISS-VSR on the one hand and REPMC with the exact model on the other. In the next subsections, we will give sufficient conditions for an approximate model to be REPMC with the exact model and show that these conditions are not restrictive.

Sufficient Lyapunov-type checkable conditions for 
SE-ISS-VSR of a discrete-time model
are presented in Theorem~\ref{adaptado}.
The proof is obtained by performing 
minor changes to the proof of the S-ISS-VSR characterization in
\cite[Theorem~3.2]{VALLARELLA201860} and is given
in \ref{app:proof:elcor}.

\begin{theorem} [Adapted from Theorem~3.2 of \cite{VALLARELLA201860}]
\label{adaptado}
Suppose that
\begin{enumerate}[i)]
    \item There exists $\mathring T > 0$ so that $\bar F (0, 0, T ) = 0$ for all 
    $T \in (0, \mathring T)$.\label{adaptado_1}
    \item There exists  $\hat T > 0$ such that for every $\epsilon >0$
    there exists $\delta = \delta(\epsilon)>0$ such that $|\bar F (x, e, T )| < \epsilon$ whenever $|x| \leq \delta$, $|e| \leq \delta$ and $T \in (0, \hat T )$. \label{adaptado_2}
    \item For every $M,E \geq 0$, there exist $C = C(M, E) >
0$ and $\check T = \check T (M, E) > 0$, with $C(\cdot, \cdot)$ nondecreasing in
each variable and  $\check T (\cdot, \cdot)$ nonincreasing in each variable,
such that $|\bar F (x, e, T )| \leq C$ for all $|x|<M$, $|e|<E$ and 
$ \check T \in (0,T)$. \label{adaptado_3}
\item   \label{adaptado_4}
There exist $\alpha_1,\alpha_2, \alpha_3  \in \K_\infty$ defined as
$\alpha_1(s):= K_1 s^N$,
$\alpha_2(s):= K_2 s^N$
and
$\alpha_3(s):= K_3 s^N$ with $N>0$ and $K_i\geq 1$ for all $i\in\{1,2,3\}$ and $\rho \in \K$ such that for
every $M,E\geq 0$ there exist $ \tilde T = \tilde T (M, E) > 0$
and $V = V_{M,E} : \R^n \rightarrow \R_{\geq0} \bigcup \{\infty\}$ such that
\begin{subequations}
\begin{align}
\alpha_1(|x|) \leq V(x), \quad \forall x \in \R^n,  \\
V(x) \leq \alpha_2(|x|), \quad \forall |x| \leq M, 
\end{align}
\end{subequations}
and
\begin{equation}
V(\bar F(x,e,T))-V(x) \leq  -T \alpha_3(|x|)
\end{equation}
for all $\rho(|e|)\leq |x| \leq M$, $|e|\leq E$ and
$T \in (0,\tilde T)$.
\end{enumerate}
then the system \eqref{eq:system1} is SE-ISS-VSR.
\end{theorem}

\subsection{Sufficient conditions for REPMC}
\label{subsec:suf-REPMC}

In Lemma~\ref{lemma:REPMC_SUFFICIENT}, we prove that REPC is
a sufficient condition for 
REPMC.
Whether REPC is also necessary for REPMC remains as an open problem.
We additionaly show that
REPC is not a restrictive condition by proving
in Theorem~\ref{lemma:REPMC_SUFFICIENT_APROX}
that 
any explicit and consistent Runge-Kutta model is REPC with the exact discrete-time model.

\begin{lemma}
\label{lemma:REPMC_SUFFICIENT}
Suppose that the pair $(\bar F^a, \bar F^b)$
is REPC, then the pair is REPMC. 
\end{lemma}

\begin{proof}
Let $M,E\geq0$ and $\T,\eta>0$ be given.
Let $|x^a|,|x^b|\leq M$ be such that $|x^a-x^b|\leq \delta$.
Let $M, E \geq0$
generate $K,T^i>0$ and $\rho \in \K_\infty$
according to Definition~\ref{def:HOA}. 
Define 
\begin{equation*}
    T^*:=\min \left \{T^i,
    \rho^{-1}\left(\frac{\eta}{e^{K\T}\T}\right),
    \rho^{-1}\left(\frac{1}{e^{K\T}\T}\right)
    \right \}
\end{equation*}
and $\alpha(\delta,T):= (1+KT)\delta+ T\rho(T) (M+  \phi(E))
$,
then \eqref{eq:orwell} is satisfied.
For all $k\in \N$ such that 
$\sum_{i=0}^{k-1} T_i \leq \T$ with $\{T_i\}\in \Phi(T^*)$ we have 
\begin{align}
&\alpha^k(0,\{T_i\}) =  \notag \\
&=  \left(\sum_{j=0}^{k-2} T_j \rho(T_j)
\prod_{i=j+1}^{k-1} (1+K T_i)+T_{k-1}\rho(T_{k-1}) \right)(M+  \phi(E)) \notag \\
 &\leq \left(\sum_{j=0}^{k-2} T_j \rho(T_j)
 e^{K \sum_{i=j+1}^{k-1}T_i}+T_{k-1}\rho(T_{k-1}) \right)
 (M+  \phi(E))
 \notag \\
 &\leq
  \rho(T^*)e^{K \T}\left(\sum_{j=0}^{k-2} T_j 
+T_{k-1} \right)(M+  \phi(E)) \notag \\
 &\leq
  \rho(T^*)e^{K \T} \T M + \rho(T^*)e^{K \T} \T  \phi(E) \leq \eta M + \phi(E).
  \label{eq:sakamoto}
\end{align}
This concludes the proof.
\end{proof}

Next, we derive a bound for the mismatch between 
the exact model and the Euler approximate model. 
Lemma~\ref{lema:21} is used in the proof of 
Theorem~\ref{lemma:REPMC_SUFFICIENT_APROX} and its proof is given in \ref{app:proof:lema:21}.
\begin{lemma}
\label{lema:21}
Suppose that Assumptions~\ref{def:UniLipschitz} and~\ref{def:bounded} hold.
Then, for every compact sets 
$\X \subset \R^n$ and $\U \subset \R^m$
there exist constants $\bar T=\bar T(\X,\U)>0$ and $\bar L=\bar L(\X,\U)>0$ such that
\begin{equation}
\left|F^e(\xi,u,T)-F^{Euler}(\xi,u,T)\right|\leq \bar L T^2 \left|f(\xi,u)\right|
\label{eq:desi1}
\end{equation}
for all $\xi \in \X$, $u\in \U$ and $T\in(0,\bar T)$.
\end{lemma}
An $s$-stage explicit Runge-Kutta model for (\ref{eq:cs}) is given by %
\begin{align}
y_1 = x, \quad
y_i &= x+T \sum_{j=1}^{i-1} a_{ij} f(y_j,u), \quad i=2,\hdots,s,  \label{charl}   \\
F^{RK}(x,u,T) &:=x+T \sum_{i=1}^{s} b_{i} f(y_i,u), \notag 
\end{align}
with $a_{ij},b_i \in \R$ for all required values of $i$ and $j$.
The Runge-Kutta model is said to be consistent if $\sum_{i=1}^s b_i =1$ \cite[Sec~3.2]{StuartDynSysAndNumAn}. %

\begin{theorem}
\label{lemma:REPMC_SUFFICIENT_APROX}
Consider that system \eqref{eq:cs} is fed back, under ZOH and possible nonuniform sampling, with the control law $U(x,e,T)$,
yielding the exact discrete-time model $\bar F^e(x,e,T)=F^e(x,U(x,e,T),T)$.
Let Assumptions~\ref{def:UniLipschitz}, \ref{def:bounded} and \ref{def:controllaw} hold and suppose that
\begin{enumerate}[i)]
    \item   \label{eq:fen0} 
    there exists $ \phi \in \K_\infty$ such that 
    for every $E\geq 0$ there exists $T^{i}:=T^{i}(E)>0$ such that
    for all $|e|\leq E$ and $T\in(0,T^{i})$ we have
    \begin{equation}
        |f(0,U(0,e,T))| \leq  \phi(|e|); 
        \label{eq:can}
    \end{equation}
    \item \label{eq:condicion} for every $M,E\geq 0$ there exist $K:=K(M,E)>0$ and $T^{ii}:=T^{ii}(M,E)>0$,
    with $K(\cdot,\cdot)$ nondecreasing in each variable and $T^{v}(\cdot,\cdot)$
    nonincreasing in each variable,
    such that for all $|x^a|,|x^b|\leq M$, $|e| \leq E$ and $T\in (0,T^{ii})$ we have
    \begin{align}
        |f(x^a,U(x^a,e,T))-f(x^b,U(x^b,e,T)) | \leq K |x^a-x^b|. 
        \label{eq:desigualdad}
    \end{align}
\end{enumerate}
Let $F^{\text{RK}}$ denote any explicit Runge-Kutta 
model for \eqref{eq:cs} and $\bar F^{\text{RK}}$ the corresponding closed-loop model involving $U(x,e,T)$. Then,
$(\bar F^{\text{RK}},\bar F^e)$ is REPC.
\end{theorem}
The proof of Theorem~\ref{lemma:REPMC_SUFFICIENT_APROX} is given in~\ref{app:proof:lem:repmcsufapp}. Theorem~\ref{lemma:REPMC_SUFFICIENT_APROX} gives sufficient conditions 
for REPC (and, via Lemma~\ref{lemma:REPMC_SUFFICIENT}, also for REPMC) between 
any explicit and consistent Runge-Kutta
model and the exact model, based on conditions on the
continuous-time plant and on the control law.
Assumptions~\ref{def:UniLipschitz} to~\ref{def:controllaw} and condition \ref{eq:fen0}) consist in mild boundedness and continuity requirements. 
Condition~\ref{eq:condicion}) is also a type of continuity requirement and allows to ensure uniqueness of solutions of the closed-loop
continuous-time model. 
Given that REPC is transitive 
it is evident that under the assumptions of Theorem~\ref{lemma:REPMC_SUFFICIENT_APROX} all explicit and consistent
Runge-Kutta models are also REPC with each other.
In particular, since the Euler model is the simplest explicit Runge-Kutta model, we have
that $(\bar F^{Euler}, \bar F^e)$ is REPC.

\begin{remark}
Theorem~\ref{lemma:REPMC_SUFFICIENT_APROX}
does not explicitly require differentiability
of the function $f$ that defines the continuous\-/time plant but
only Lipschitz-type conditions. The latter conditions may imply almost-everywhere differentiability but only of first order. Therefore, the requirements imposed by Theorem~\ref{lemma:REPMC_SUFFICIENT_APROX} on $f$ are weaker than the high-order differentiability required to ensure convergence of a high-order Runge-Kutta model.
\end{remark}

\subsection{Intersample bound}
\label{app:inter}
Once any type of S-ISS-VSR property (e.g. SE-ISS-VSR) is established 
for the exact discrete-time closed-loop  model 
we can then derive a bound for the intersample behaviour for the sampled-data system.

\begin{lemma}
\label{lema_interbound}
Consider that $x_{k+1}= F^e(x_k,U(x_k,e_k,T_k),T_k)$ is S-ISS-VSR with functions $\beta \in \KL$ and $\gamma \in \K_\infty$. Then, for 
the closed-loop sampled-data system given by \eqref{eq:cs} and $U(x_k,e_k,T_k)$ under ZOH we have
\begin{align}
\label{RRRR5}
\left|x(t_k+t)\right| %
&\leq \beta\left(|x_0|,\sum_{i=0}^{k-1}T_i\right)+\gamma\left(\sup_{0\le i \le k-1}|e_i| \right )+ C(R,E)t.
\end{align}
for all $t \in [0, T_k]$, $\{T_i\} \in \Phi(T^\star)$, $|x_0|\leq M$ and $\|\{e_i\}\| \leq E$. 
If additionally, the control law is independent of the current sampling period, i.e. $U(x_k,e_k)$, we have
\begin{equation}
    |x(t)| \leq \beta\left(|x_0|,t\right)+\gamma\left(\sup_{0\le i \le k-1}|e_i|\right )  
    \label{eq:R6}
\end{equation}
for all $t\geq 0$, $\{T_i\} \in \Phi(T^\star)$, $|x_0|\leq M$ and $\|\{e_i\}\| \leq E$. 
\end{lemma}
For the particular case where  the control action
is independent of the current sampling period,  
such as in the emulation case,
the evolution between consecutive samples $x_k$ and $x_{k+1}$ is determined by $x(t_k+t)=F^e(x_k,U(x_k,e_k),t)$ for all $t\in [0,T_{k}]$.
The values that $x(t)$ takes in the interval $[t_k,t_{k+1}]$
are given by the open-loop exact discrete-time model under the constant input $U(x_k,e_k)$, regardless of the value of $t_{k+1} = t_k + T_k$.
Given that 
the bound \eqref{eq:S-ISS-VSR} holds for every possible sequence of sampling periods, 
it then straightforwardly follows that the 
bound that is ensured for the exact discrete-time model
also holds for the sampled-data system with the same functions $\beta \in \KL$,
$\gamma \in \K_\infty$ and maximum admissible sampling period $T^\star$.
An intersample bound \eqref{RRRR5} leads to a bound of the form \eqref{eq:R6}
with different functions $\beta$ and $\gamma$ by repeating the bound \eqref{eq:R6} and considering each sampling instant as a new initial time.
However, the precise way in which the maximum admissible sampling period depends on such a decreasing intersample bound is not at all straightforward and the derivation of the functions for the general case is thus left for future work. The proof of Lemma \ref{lema_interbound} is in \ref{app_lema_interbound}.

\section{Example}
\label{sec:example}
Consider the
continuous-time plant $\dot x= f(x,u)=x^3+u$ in 
Example~A of \cite{IEEETAC18}.
Note that the solution of the open-loop continuous-time plant
may not exist for all times due to finite escape time. 
We next perform discrete-time design based on the use of Runge-Kutta models.
We consider the Euler model of the open-loop plant
\begin{equation}
F^{Euler}(x,u,T):=x+T (x^3+u).   \label{C1} 
\end{equation}
We also consider a desired closed-loop continuous-time stable system, e.g. 
$\dot x = f_d(x)= -x^3-2x$,
which we approximate by means of a second-order Runge-Kutta model,
namely the Heun model:
\begin{align}
\bar F^{Heun}&(x,T) :=x+\frac{T}{2} (f_d(x)+f_d(x+T f_d(x))) \notag \\ 
&= \frac{1}{2}(x^3 + 2x)^3T^4  -\frac{3}{2}x(x^3 + 2x)^2T^3 
 \label{C2}  \\
&+ (2x + (3x^2(x^3 + 2x))/2 + x^3)T^2 - ( x^3 + 2x)T + x. \notag
\end{align}
Matching equations \eqref{C1} and \eqref{C2} and solving for $u$ we obtain
the following control law
for the disturbance-free case
\begin{equation}
u=U(x,0,T) :=\frac{\bar F^{Heun}(x,T) - x}{T} -x^3.\label{A3}
\end{equation}
Thus, we ensure that the behaviour of the Euler closed-loop model $x_{k+1}=\bar F^{Euler}(x_k,0,T_k)= F^{Euler}(x_k,U(x_k,0,T_k),T_k)$
is described by \eqref{C2}.
If we consider additive state-measurement errors, this yields
\begin{subequations}
     \label{C5C}
    \begin{align}
\bar F^{Euler}(x,e,T)&:=x+T (x^3+U(x,e,T))   \label{C5}  \\
U(x,e,T)&=\frac{\bar F^{Heun}(x+e,T) - x-e}{T} -(x+e)^3 \label{D3}.
    \end{align}
\end{subequations}
We will prove that \eqref{C5C} is SE-ISS-VSR
via Theorem \ref{adaptado}.
The continuity and boundedness assumptions \ref{adaptado_1}), \ref{adaptado_2}) and \ref{adaptado_3}) of Theorem \ref{adaptado} 
are easy to verify for \eqref{C5C}.
Now we will prove assumption \ref{adaptado_4}).
Define $\alpha_1, \alpha_2, \alpha_3 \in \K_\infty$ 
via $\alpha_1(s)=\alpha_2(s)=s^2$,
$\alpha_3(s)=s^2/2$
and $\rho(s) = s/K$ with $K > 0$ to be selected.
Let $M \geq 0$ and $E \geq 0$ be
given and define $V(x) = x^2$.
We have
\begin{align}
&V(\bar F(x,e,T))-V(x)= \sum_{i=1}^8
a_i(x,e) T^i  \notag \\
&= \left[a_1(x,e)+ T\sum_{i=2}^8 a_i(x,e) T^{i-2}\right]T \leq \left[a_1(x,e)+ T\sum_{i=2}^8 |a_i(x,e)| \right]T \notag 
\end{align}
for all $T \in (0,1)$,
where each $a_i(x,e)$
is a multivariate polynomial
in the indeterminates $x, e$, and
$a_1(x,e)=- 4e^3x - 12e^2x^2 - 12ex^3 - 4ex - 2x^4 - 4x^2$.
Selecting $K=0.01$, noting that whenever $\rho(|e|)\leq |x|$  we have
$|e|\leq K |x|$ and taking absolute values on sign-indefinite terms
of $a_1(x,e)$ 
we can bound it as
\begin{align}
a_1(x,e)&\leq 4K^3 x^4 + 12K^2 x^4 + 12K x^4 + 4Kx^2 - 2x^4 - 4x^2\notag  \\ 
&\leq (-4+4K)x^2+(-2+28K)x^4 \leq -1.5 (x^2+x^4). \notag  
\end{align}
Defining $C=670$,
replacing the negative definite terms of each $a_i(x,e)$ by zero,
taking absolute values on sign-indefinite terms and
bounding each $e$ according to $|e|\leq K|x|$ we obtain 
$\sum_{i=2}^8 |a_i(x,e)|  \leq C \sum_{i=1}^{9}x^{2i}$
for all $\rho(|e|)\leq |x|\leq M$, $|e|\leq E$ and $T\in(0,1)$.
Select
$\tilde T= \min \left\{ \frac{1}{C \left(1+\sum_{i=1}^{8}M^{2i}\right)},1 \right\},
$
then
\begin{align*}
&a_1(x,e)+ T\sum_{i=2}^8 |a_i(x,e)| \leq 
-1.5 (x^2+x^4) +TC\sum_{i=1}^{9}x^{2i}  \notag \\
&\leq -0.5 x^2 +\left(TC \left(1+\sum_{i=1}^{8}M^{2i}\right)-1 \right)x^2 \leq -0.5x^2 = -\alpha_3 (|x|).
\end{align*}
for all $\rho(|e|)\leq |x|\leq M$, $|e|\leq E$
and $T\in(0, \tilde T)$.  

Next, we will prove that \eqref{C5C} is not globally stable. 
Suppose that there exists $T^* \in (0,1)$ such that 
the system $x_{k+1}=\bar F^{Euler}(x_k,0,T_k)$ is globally exponentially stable under VSR for all 
$\{T_i\} \in \Phi(T^*)$. Define $\bar T:=T^*/2$ and
consider the constant sequence $\{T_i\}\in \Phi(T^*)$, with $T_i = \bar T$ for all $i$.
For all $|x|>\sqrt{(3 -  \bar T)/(2\bar T)}$ we have
\begin{align*}
&|\bar F^{Euler}(x,0, \bar T)|= |x+\bar T(x^3+U(x,0, \bar T))| \notag \\
&= \left|\frac{\bar T}{2} \left[2x^2+1\right]\left[ \bar T(2x^2+1)-2\right]\left[2\bar T x^2(2 \bar T x^2+T-1)\right]+1 \right||x| \notag \\
&> \left|\frac{\bar T}{2} \left[2x^2+1\right]\left[4\bar T x^2\right]+1 \right||x|> |x|. 
\end{align*}
The solution thus diverges for large values of the state. 

Next, by means of Theorem~\ref{theorem:SEISS} we will prove that 
$\bar F^e$ exhibits the stronger SE-ISS-VSR property  
that could not be ensured by the existing results.
Theorem~2 of \citet{nevsic2009stability} cannot be applied
to prove asymptotic stability of $\bar F^e$ even in the absence of errors
due to the fact that the Euler model \eqref{C5C} is not globally stable.
Theorem~1 of \citet{IEEETAC18} can be applied
but only to ensure semiglobal practical (not asymptotic) ISS-VSR.

\begin{figure}[ht]
\vspace{-0.4cm}
\includegraphics[width=0.49\textwidth]{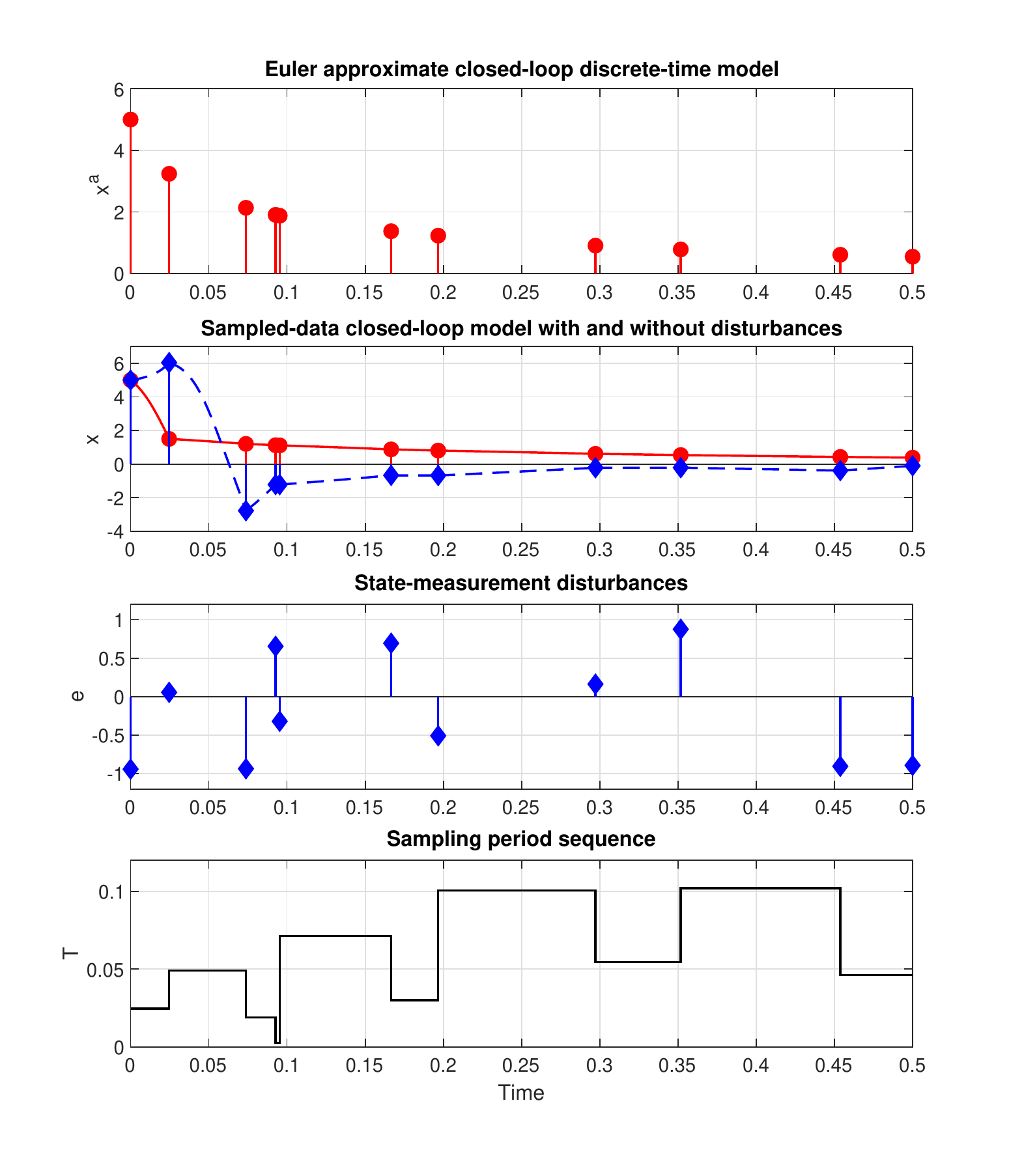}
\vspace*{-11mm}
\caption{
Evolution of the sampled-data system for the disturbance-free case (red solid line) and
in the presence of state-measurement disturbances (blue dashed line) for a given sequence of sampling periods. The state of the approximate and exact discrete-time models is represented with stems in both cases.}
\end{figure}

First, we prove that 
$(\bar F^{Euler},\bar F^e)$ is REPMC via 
Theorem~\ref{lemma:REPMC_SUFFICIENT_APROX}.
Assumptions~\ref{def:UniLipschitz} to~\ref{def:controllaw}
are easy 
to verify for the 
plant $\dot x =f(x,u)$ and control law $U(x,e,T)$.
To prove conditions \ref{eq:fen0}) and \ref{eq:condicion}), 
define $\phi \in \K_\infty$ via $\phi(s):=s^9 + 3s^7 + 3s^5$, then
\begin{equation}
\left|f(0,U(0,e,T))\right|= \frac{1}{2}\left|e^9 + 3e^7+ 3e^5
\right|\leq|e|^9 + 3|e|^7 + 3|e|^5 = \phi(|e|) \notag
\end{equation}
for all $T \in (0,1)$,
thus  \ref{eq:fen0}) holds.
The function $f(x,U(x,e,T))$ is easily seen to be a multivariate polynomial in the variables $x, e, T$. Therefore, this function is locally Lipschitz in $x$, uniformly with respect to the other variables in compact sets and \ref{eq:condicion}) holds.
By Theorem~\ref{lemma:REPMC_SUFFICIENT_APROX},
$(\bar F^{Euler},\bar F^e)$ is REPC and by Lemma~\ref{lemma:REPMC_SUFFICIENT} also REPMC. 
By Theorem~\ref{theorem:SEISS} then 
$\bar F^e$
is SE-ISS-VSR.

In order to illustrate the results we simulated
the approximate Euler closed-loop model (used for control
design) and the original sampled-data model (both with and without disturbances) from initial 
condition $x(0)=5$ for the same
sequence of random sampling periods on a given interval.
We considered the case where random continuous uniformly distributed state-measurment disturbances $e \in [-1,1]$ are present.
The simulations in Figure 1 show the expected behaviour.

\section{Conclusions}
\label{sec:conclusions}

We have presented novel results that
guarantee the semiglobal exponential input-to-state stability (SE-ISS-VSR)
for discrete-time models of nonlinear nonuniformly sampled plants
under state-measurement or actuation-error disturbances
based on approximate discrete-time models.
We have proved that
under a multistep consistency property (REPMC)
between two discrete-time models the
SE-ISS-VSR property is carried
over between models.
We have shown that a much easier-to-verify one-step condition (REPC)
is a transitive property
and that it constitutes
a sufficient condition for REPMC.
Furthermore, we have proved that 
under mild boundedness and continuity conditions on the continuous-time model and the control law,
any explicit and consistent Runge-Kutta (approximate) model is REPC with the exact discrete\-/time model
and thus can be used for control design.
We have provided an example of semiglobal exponential stabilization
discrete-time design based on Runge-Kutta models.

Very recently, we proved that the simplest \emph{implicit} Runge-Kutta model (backward Euler), 
is also REPC with the exact model \citep{AADECA20}.
We conjecture that this holds also for all
implicit Runge-Kutta models and also for other types of well-known models. This is a topic for future work, as well as extending Theorem \ref{lemma:REPMC_SUFFICIENT_APROX} to the case of dynamic controllers.

\appendix


\section{Proof of Proposition~\ref{cor:1}}
\label{app:proof:cor:1}
Let the REPC property 
define  $\phi_1 \in \K_\infty$
and $\phi_2 \in \K_\infty$
for the  
the pairs  $(\bar F^a, \bar F^b)$ and $(\bar F^b, \bar F^c)$,
respectively.
Suppose $M,E \geq 0$ given
and let them generate
$K_1, T^{*,1}>0$, $\rho_1 \in \K_\infty$
and $K_2, T^{*,2}>0$, $\rho_2 \in \K_\infty$ according to Definition~\ref{def:HOA}
for the pairs  $(\bar F^a, \bar F^b)$ and $(\bar F^b, \bar F^c)$, respectively.
Consider $|x^a|,|x^c|\leq M$ and $|e|\leq E$ given.
Define
$K:=K_1$,
$T^*:=\min \{T^{*,1},T^{*,2}\}$
and $\rho,\phi \in \K_\infty$ via
$\rho:=\rho_1+\rho_2$ and
$\phi:=\phi_1+ \phi_2$.
Thus we have
\begin{align}
&\left|\bar F^a(x^a,e,T) - \bar F^c(x^c,e,T)\right| \notag \\
&\leq
\left|\bar F^a(x^a,e,T) - \bar F^b(x^c,e,T)\right|
+\left|\bar F^b(x^c,e,T) - \bar F^c(x^c,e,T)\right| \notag  \\
&\leq 
(1+K_1T) |x^a-x^c|+ T \rho_1(T) \left(\max\{|x^a|,|x^c|\}+\phi_1(|e|) \right) \notag \\
& \qquad +(1+K_2T)|x^c-x^c|+ T \rho_2(T) \left(\max\{|x^c|,|x^c|\}+\phi_2(|e|) \right) \notag \\
&\leq 
(1+KT) |x^a-x^c|+ T \rho_1(T) \left(\max\{|x^a|,|x^c|\}+\phi(|e|) \right) \notag \\
& \qquad+ T \rho_2(T) \left(\max\{|x^a|,|x^c|\}+\phi(|e|) \right) \notag \\
&\leq(1+KT)(|x^a-x^c|)  
+ T \rho(T)(\max\{|x^a|,|x^c|\}+ \phi(|e|)).
\label{eq:eded}
\end{align}
for all $|x^a|,|x^c|\leq M$, $|e|\leq E$ and $T\in(0,T^*)$
and  the pair $(\bar F^a, \bar F^c)$ is REPC.

\qed

\section{Proof of Lemma~\ref{lem:menem}}
\label{app:proof:menem}

Consider $M_a,E\geq0$ and
$\T,\eta > 0$ given.
Since 
$\bar F^a$
is REPMC with 
$\bar F^e$
define $M_e:= (1+\eta) M_a+ \phi(E)$ and
generate $T^*:=T^*(M_e,E,\T,\eta)>0$
and function $\alpha:\R_{\geq0} \times \R_{\geq0}
 \rightarrow \R_{\geq0} \cup \{\infty\}$
according to Definition~\ref{eq:default}.
Define $T^L:=T^*$ 
and consider
$\{T_i\} \in \Phi(T^L)$ and
$\|\{e_i\}\|\leq E$.
Define 
$\Delta x_k:=x^e(k,\xi,\{e_i\},\{T_i\})-x^a(k,\xi,\{e_i\},\{T_i\})$.
For $k=0$ we have 
\begin{equation}
|\Delta x_0|:=|\xi-\xi|=0
\leq \eta |\xi|+ \phi(E).
\end{equation}
We proceed by induction on $k$.
Let $k\in \N_0$
be such that $\sum_{i=0}^{k-1} T_i
\in [0,\T]$.
Suppose that 
$|x^a(j,\xi,\{e_i\},\{T_i\})|\leq M_a$ and 
$|\Delta x_j|\leq \eta |\xi|
+\phi(\sup_{0\leq i \leq j-1} |e_i|)$ for all $0\leq j\leq k$.
Thus $|x^e(j,\xi,\{e_i\},\{T_i\})| \leq M_a+\eta |\xi|+\phi(\sup_{0\leq i \leq j-1} |e_i|) \leq M_e$ for all $0\leq j\leq k$.
From \eqref{eq:orwell} and \eqref{eq:trap} 
and noting that by causality 
$\Delta x_{k+1}$ cannot depend on future
values of $e_i$
we have
\begin{align}
&|\Delta x_{k+1}|=|\bar F^e(x^e_k,e_k,T_k)-\bar F^a(x^a_k,e_k,T_k)| \leq \alpha(|\Delta x_k|,\{T_i\}) \notag \\
&\leq \alpha^{k+1}(|\Delta x_0|,\{T_i\}) =\alpha^{k+1}(0,\{T_i\}) \leq \eta |\xi|+\phi\left(\sup_{0 \leq i \leq k} |e_i|\right).\tag*{\qed}
\end{align}

\section{\textbf{Proof of Theorem~\ref{theorem:SEISS}}}
\label{app:proof:SEISS}
Let $K_a,\lambda_a>0$, $\gamma_a \in \K_\infty$ and $T^\star(\cdot,\cdot)$
characterize the SE-ISS-VSR property of 
$x^a_k=\bar F^a(x^a,e_k,T_k)$.
Consider $M\geq 0$ and $E\geq 0$ given.
Consider $\phi \in \K_\infty$ from \ref{item:hoa123333})
and define $\hat \gamma \in \K_\infty$ via
$\hat \gamma := \gamma_a + \phi$.
Let $\delta \in(0,1)$ and  $\eta\in(0,\delta)$. 
Define $M_a:= K_aM+\frac{1}{1-\delta}\hat \gamma(E)$.
Let $\T:=\frac{1}{\lambda_a} \ln{\frac{K_a}{\delta-\eta}}$.
Define $\T_1:=\T+1$ and generate
$T^L:= T^L(M_a,E,\T_1,\eta)$
according to Lemma~\ref{lem:menem}.
We have
\begin{equation}
|x^a_k|\leq  K_a |\xi| \exp{\left(-\lambda_a \sum_{i=0}^{k-1} T_i\right)}
+ \gamma_a \left( \sup_{0\leq i\leq k-1} |e_i| \right)
\end{equation}
for all $k\in \N_0$, $|\xi| \leq M_a$,
$\{T_i\} \in \Phi(T^\star(M_a, E))$ and
$\|\{e_i\}\| \leq E$.
Define
$\bar T<\min\left\{1,T^\star(M_a,E),
  T^L \right\}$.
Consider $\{T_i\} \in \Phi(\bar T)$.
For every $k\in\N_0$ and $j\in \N$, define
\begin{align}
s(k)&:= \sup \left\{r\in \N_0 : r \ge k+1, \sum_{i=k}^{r-1} T_i \leq \T_1 \right\} \text{ and} \\
s^j(k)&:=\overbrace{s(\hdots s(s}^j(k))) 
\end{align}
Note that $s^1(k) \ge k+1$ for all $k\in\N_0$ because $\T_1>1$ and $T_i < 1$ for all $i\in\N_0$. Also, $\sum_{i=k}^{s^1(k)-1} T_i > \T_1- \bar T > \T_1 -1= \T$ holds for all $k\in\N_0$.
For every $k,\ell \in \N_0$ with $k\geq \ell$ and 
$|\xi|\leq M_a$ define
$\Delta x^\xi_{k,\ell} :=
|x^b(k-\ell,\xi,\{e_{i+\ell}\},\{T_{i+\ell}\})
-x^a(k-\ell,\xi,\{e_{i+\ell}\},\{T_{i+\ell}\})|$.
Consider that $|\xi|\leq M$.
Then $|x^a_k| \leq M_a$ for all $k\in\N_0$,
$\{T_i\} \in \Phi (\bar T)$ and $\|\{e_i\}\|\leq E$.
From \ref{item:hoa123333}),
according to Lemma~\ref{lem:menem}, for all $k\in \N_0$
for which $\sum_{i=0}^{k-1}\leq \T_1$ and 
$\{T_i\} \in \Phi(\bar T)$ we have
\begin{align*}
|x^b&(k,\xi,x\{e_i\},\{T_i\})| 
 \leq |x^a(k,\xi,\{e_i\},\{T_i\})|+ \left|\Delta x^\xi_{k,0} \right| \notag \\
& \leq  K_a |\xi| \exp{\left(-\lambda_a \sum_{i=0}^{k-1} T_i\right)}
+ \gamma_a \left( \sup_{0\leq i\leq k-1} |e_i| \right) + \left|\Delta x^\xi_{k,0} \right| \notag \\
& \leq  K_a |\xi| \exp{\left(-\lambda_a \sum_{i=0}^{k-1} T_i\right)}
+ \gamma_a \left( \sup_{0\leq i\leq k-1} |e_i| \right) + \eta|\xi|+\phi
\left(\sup_{0 \leq i \leq k-1 } |e_i|\right)   \notag \\
& \leq  K_a |\xi| \left( \exp{\left(-\lambda_a \sum_{i=0}^{k-1} T_i\right)}
+ \frac{\eta}{K_a}\right) + \hat \gamma \left( \sup_{0\leq i\leq k-1} |e_i| \right).  \notag 
\end{align*}
For the sake of notation define $x^b_{s^j(0)}:=x^b(s^j(0),\xi,\{e_i\},\{T_i\})$.
For instant $s^1(0)$ we have
\begin{align}
|x^b_{s^1(0)}|
& \leq  K_a |\xi| \left( \exp{\left(-\lambda_a \sum_{i=0}^{s^1(0)-1} T_i\right)}
+ \frac{\eta}{K_a}\right) + \hat \gamma \left( \sup_{0\leq i\leq s^1(0)-1} |e_i| \right) \notag \\
& \leq  K_a |\xi| \left( \exp{\left(-\lambda_a  \T\right)}
+ \frac{\eta}{K_a}\right) + \hat \gamma \left( \sup_{0\leq i\leq s^1(0)-1} |e_i| \right)  \notag \\
& \leq  K_a |\xi| \left(\frac{\delta-\eta}{K_a}+\frac{\eta}{K_a}\right) +
\hat \gamma \left( \sup_{0\leq i\leq s^1(0)-1} |e_i| \right)  \notag \\
& \leq  \delta|\xi| +\hat \gamma \left(  \sup_{0\leq i\leq s^1(0)-1} |e_i| \right) \label{eq:pink_floyd}  \\  
&\leq  \delta M + \hat \gamma \left( E \right) \leq M_a
\notag
\end{align}
Note that for an initial condition such that
$|x^b_{s^j(0)}|\leq M_a$ for some $j\in \N_0$ then,
following the same reasoning that leads to \eqref{eq:pink_floyd},
we can bound $|x^b_{s^{j+1}(0)}|$
as 
\begin{equation}
|x^b_{s^{j+1}(0)}| \leq \delta |x^b_{s^{j}(0)}|+ \hat \gamma \left( E \right).     
\end{equation}
Thus, we have
$|x^b_{s^{j+1}(0)}| \leq \delta M_a+\hat \gamma (E) = \delta (K_aM+ \frac{1}{1-\delta} \hat \gamma(E))+ \hat \gamma(E) =
\delta K_aM+ (\frac{\delta}{1-\delta}+1) \hat \gamma(E))=\delta K_aM+ \frac{1}{1- \delta}\hat \gamma(E) < M_a$.
Thus $\left|x^b_{s^j(0)} \right| \leq M_a$ for all $j\in \N_0$.
Then we can apply \eqref{eq:pink_floyd} iteratively to obtain
\begin{align}
\left|x^b_{s^j(0)} \right| 
&\leq \delta^j |\xi|+\hat \gamma\left( \sup_{0\leq i\leq s^j(0)-1} |e_i| \right)\sum_{i=0}^{j-1} \delta^i 
 \notag \\
&= 
\exp{\left(- \lambda_1 j\right)} |\xi|+ \frac{1- \delta^{j}}{1-\delta} \hat \gamma
\left( \sup_{0\leq i\leq s^j(0)-1} |e_i| \right)
\label{eq:883}
\end{align}
where $ \lambda_1:= \ln{\frac{1}{\delta}}>0$. 
Using the definition of $s^j(0)$
we have
\begin{equation}
    -  \lambda_1 j = -  \lambda_1 j \frac{\T_1}{\T_1} \leq 
    -  \lambda_1 \frac{\sum_{i=0}^{s^j(0)-1} T_i}{\T_1}
    = -\bar \lambda \sum_{i=0}^{s^j(0)-1} T_i
    \label{eq:fin}
\end{equation}
where $\bar \lambda :=  \lambda_1/ \T_1$.
Using \eqref{eq:fin} on \eqref{eq:883} 
and the fact that for $\delta \in(0,1)$ it holds that
$\frac{1- \delta^{j}}{1-\delta} \leq
\frac{1}{1-\delta}$ for all $j\in N_0$,
we have
\begin{equation}
\left|x^e_{s^j(0)} \right|  \leq  \exp{\left(-\bar \lambda \sum_{i=0}^{s^j(0)-1} T_i\right)} |\xi|
+ \frac{1}{1-\delta} \hat  \gamma\left( \sup_{0\leq i\leq s^j(0)-1} |e_i| \right).
\label{eq:alx3}
\end{equation}
Define $x^b_k(\xi):=x^b(k,\xi,\{e_i\},\{T_i\})$.
From \ref{item:SES12}) and
Lemma~\ref{lem:menem},
for all $k\in[s^j(0),s^{j+1}(0)]$, we have that
\begin{align}
&\left|x^b_k(\xi)\right|  =\left|x^b(k-s^j(0),x^b_{s^j(0)},\{e_{i+s^j(0)}\},\{T_{i+s^j(0)}\})\right| \notag \\
& \leq \left|x^a(k-s^j(0),x^b_{s^j(0)},\{e_{i+s^j(0)}\},\{T_{i+s^j(0)}\}) \right| \notag \\
&\qquad +\big|x^b(k-s^j(0),x^b_{s^j(0)},\{e_{i+s^j(0)}\},\{T_{i+s^j(0)}\}) \notag \\ &\qquad -x^a(k-s^j(0),x^b_{s^j(0)},\{e_{i+s^j(0)}\},\{T_{i+s^j(0)}\}) \big| \notag \\
& \leq  K_a\left|x^b_{s^j(0)}\right|+ \gamma_a
\left( \sup_{s^j(0)\leq i\leq k-1} |e_i| \right) + \eta \left|x^b_{s^j(0)}\right|+ \phi\left( \sup_{s^j(0)\leq i\leq k-1} |e_i| \right) \notag \\
& \leq  (K_a+ \eta) \left|x^b_{s^j(0)}\right|+\hat \gamma
\left( \sup_{s^j(0)\leq i\leq k-1} |e_i| \right).
\label{eq:hefes2}
\end{align}
Using \eqref{eq:alx3} and \eqref{eq:hefes2},
for all $k\in[s^j(0),s^{j+1}(0)]$, we have 
\begin{align}
&|x^b_k(\xi)| 
\leq (K_a+ \eta) \left|x^b_{s^j(0)}\right|+ \hat \gamma
\left( \sup_{s^j(0)\leq i\leq k-1} |e_i| \right) \notag \\
&\leq (K_a+ \eta) \left( \exp{\left(-\bar \lambda \sum_{i=0}^{s^j(0)-1} T_i\right)} |\xi|
+ \frac{1}{1-\delta} \hat\gamma
\left( \sup_{0\leq i\leq s^j(0)-1} |e_i| \right)
\right) \notag \\
& \qquad +\hat\gamma
\left( \sup_{s^j(0)\leq i\leq k-1} |e_i| \right) \notag \\
&\leq (K_a+ \eta) \exp{\left(-\bar \lambda \sum_{i=0}^{s^j(0)-1} T_i\right)} |\xi| \notag  \\
&\qquad+ \left( \frac{(K_a+\eta)}{1-\delta}+1\right)
\hat\gamma 
\left( \sup_{0\leq i\leq s^j(0)-1} |e_i| \right)\notag \\
&\leq (K_a+ \eta) \exp{\left(-\bar \lambda \left(\sum_{i=0}^{k-1} T_i- \T_1 \right)\right)} |\xi| 
 + \gamma_b
\left( \sup_{0\leq i\leq s^j(0)-1} |e_i| \right)\notag \\
&\leq (K_a+ \eta) \exp{\left(\bar \lambda \T_1\right)}
\exp{\left(-\bar \lambda \sum_{i=0}^{k-1} T_i\right)} |\xi|  + \gamma_b
\left( \sup_{0\leq i\leq s^j(0)-1} |e_i| \right)\notag \\
&\leq K_b
\exp{\left(-\bar \lambda \sum_{i=0}^{k-1} T_i\right)} |\xi| +
\gamma_b
\left( \sup_{0\leq i\leq s^j(0)-1} |e_i| \right)\notag 
\end{align}
for all $k\in \N_0$, $|\xi|\leq M$ and $\{T_i\}\in \Phi(\bar T(M,E))$,
where $K_b:= (K_a+\eta)\exp{\left(\bar \lambda \T_1\right)}$
$=(K_a+\eta)/\delta$
and $\gamma_b \in \K_\infty$ is defined via 
$\gamma_b:= \left(\frac{K_a+\eta}{1-\delta}+1 \right) \hat \gamma $. \qed

\section{Proof of Theorem~\ref{adaptado}}
\label{app:proof:elcor}

The proof copies the proof of 
2. $\Rightarrow$ 1. of \cite[Theorem~3.2]{VALLARELLA201860} 
but keeps track of the changes introduced
by the fact that $\alpha_i(s) = K_i s^N$
with $N>0$ and $K_i\geq 1$ for all $i\in\{1,2,3\}$.

Since the assumptions of condition 2. of \cite[Theorem~3.2]{VALLARELLA201860} 
are satisfied, then by the latter theorem we know that system (4) is S-ISS-VSR.
The function $\alpha \in \K_\infty$
in \cite[eq.(28)]{VALLARELLA201860} results
$\alpha(s):=\alpha_3 \circ \alpha_2^{-1}(s)=\tilde \lambda s$
where  $ \tilde \lambda:=\frac{K_3}{K_2^{1/N}}$.
Therefore, the right-hand side of inequality \cite[eq.(32)]{VALLARELLA201860} is linear in $y$.
It then follows 
that the function $\beta_1 \in \KL$ 
in \cite[eq.(33)]{VALLARELLA201860} 
is given by
$\beta_1(s,t)=s\exp{\left(- \tilde \lambda t\right)}$.
Then, the function $\beta \in \KL$ in \cite[eq.(35)]{VALLARELLA201860} 
defined via $\beta(s,t):=\alpha_1^{-1}(2 \beta_1(\alpha_2(s),t))$ 
becomes $\beta(s,t)=Ks \exp{(- \lambda t)}$ where $K:= \frac{2 K_2}{K_1^{1/N}}$
and $\lambda:= \frac{\tilde \lambda }{N}$. 
Since this function $\beta$ characterizes the S-ISS-VSR property, 
it follows from Definition~\ref{def:SE-ISS-VSR} that system (4) is SE-ISS-VSR.
\qed

\section{Proof of Lemma~\ref{lema:21}}
\label{app:proof:lema:21}
Consider $\tilde \X \subset \R^n$ and $\U \subset \R^m$ given and
let $\phi_u(t,\xi):= F^e(\xi,u,t)$ be the unique solution of \eqref{eq:cs} 
that begins from initial condition $\xi \in \tilde \X$ at $t_0=0$ and has
a constant input $u\in \U$. 
Then
\begin{equation}
\label{eq:solexact}
\phi_u(t,\xi)=\xi + \int_0^t f(\phi_u(\tau,\xi),u) d\tau.
\end{equation}
Define $C_u:=\max_{u\in \U}\{|u|\}$, $R:=\max \{1,\max_{x\in\tilde \X} |x|\}$
and $\hat \X:= \{x \in \R^n:|x|\leq 2R\}$
and generate 
$C_f=C_f(2R,C_u)$
from Assumption~\ref{def:bounded}. 
Then
\begin{equation*}
|\phi_u(t,\xi)| \leq |\xi| + \int_0^t |f(\phi_u(\tau,\xi),u)| d\tau \leq R+ C_ft.
\end{equation*}
for all $t \in (0,\bar T)$
with $\bar T:=R/C_f$.
Define $\X :=\{ x:|x|\leq R+C_f \bar  T\}$ 
and 
$\tilde L:=\tilde L(\X,\U)$ 
from Assumption~\ref{def:UniLipschitz}.
The error between the solutions
of the exact model and the Euler approximate model
after one step of duration
$T \in (0, \bar  T)$
from initial condition $x^e_0=x^a_0=\xi$
with $\xi \in \tilde \X$ and
input $u\in \U$
results
\begin{align}
\vartheta(T):&=F^e(\xi,u,T)-F^{Euler}(\xi,u,T)  \notag \\
&=\xi + \left( \int_{0}^{T} f(\phi_{u}(\tau,\xi),u) d\tau \right) - \xi - T f(\xi,u) \notag \\
&=\int_{0}^{T} f(\phi_{u}(\tau,\xi),u) d\tau - T f(\xi,u) \notag \\
&=\int_{0}^{T} f(\phi_{u}(\tau,\xi),u)-  f(\xi,u) d\tau \label{eq:12345}. 
\end{align}
Taking the norm on both sides of \eqref{eq:12345}
we have
\begin{align}
&|\vartheta(T)|=\left |\int_{0}^{T} f(\phi_{u}(\tau,\xi), u)-  f(\xi,u) d\tau \right| \notag \displaybreak[0]\\
&\leq \int_{0}^{T} \left|f(\phi_{u}(\tau,\xi),u)-  f(\xi,u)\right| d\tau \leq \tilde L\int_{0}^{T} |\phi_{u}(\tau,\xi)-\xi| d\tau \notag   \\
&= \tilde L\int_{0}^{T} \left|F^e(\xi,u,\tau)-F^{Euler}(\xi,u,\tau)+\tau f(\xi,u)\right| d\tau \notag   \\
&\leq \tilde L\int_{0}^{T} |\vartheta(\tau)|d \tau+ \tilde L \int_{0}^{T} \tau | f(\xi,u)| d\tau \notag   \\
&\leq \tilde L \int_{0}^{T} |\vartheta(\tau)|d \tau+ \tilde L\frac{T^2}{2} | f(\xi,u)| \label{eq:bowie4}  
\end{align}
From \eqref{eq:bowie4}, by  Gronwall's inequality, we can bound the error as
$|\vartheta(T)| \leq \tilde L\frac{T^2}{2} | f(\xi,u)| e^{\tilde L T}$
for all $T \in [0, \bar T)$.
Defining $\bar L:= \frac{1}{2} \tilde L e^{\tilde L \bar  T}$
we have that
\begin{equation}
|F^e(\xi,u,T)-F^{Euler}(\xi,u,T)| \leq   \bar L T^2 |f(\xi,u)|
\end{equation}
for all $\xi \in \X$, $u\in \U$ and $T\in(0,\bar T)$.
\qed

\section{\textbf{Proof of Theorem~\ref{lemma:REPMC_SUFFICIENT_APROX}}}
\label{app:proof:lem:repmcsufapp}

  We will establish that $(\bar F^{RK},\bar F^e)$ is REPC by showing that both $(\bar F^{Euler},\bar F^e)$ and $(\bar F^{RK},\bar F^{Euler})$ are REPC 
and using the fact that REPC is transitive.
  
  Consider $M,E \geq 0$ given and
let them generate $K,T^{ii}>0$ from \ref{eq:condicion})
and $T^{i}>0$ from \ref{eq:fen0}).
Let $M,E$ generate $C_u,T^{u}>0$ from Assumption~\ref{def:controllaw}.
Define 
$\X:= \{x\in \R^n:|x|\leq M\}$
and $\U:= \{u\in \R^m:|u|\leq C_u\}$.
\begin{claim}
\label{clm:beta}
  $(\bar F^{Euler},\bar F^e)$ is REPC.
\end{claim}
\indent\emph{Proof of Claim~\ref{clm:beta}:}
From Assumptions~\ref{def:UniLipschitz} and \ref{def:bounded}, the conditions of Lemma~\ref{lema:21} hold. Let $\X$ and $\U$ 
generate $\bar L, \bar T>0$ from
Lemma~\ref{lema:21}, so that the open-loop condition~(\ref{eq:desi1}) holds for all $x\in\X$, $u\in\U$ and $T\in (0,\bar T)$.
Define $T^{o}:= \min \{\bar T,T^{u}, T^{i}, T^{ii}\}$,
$\bar K:=\max\{K,1\}$
and $\rho \in \K_\infty$ via $\rho(s):= \bar K \bar L s$.
For all $|x^e|,|x^a| \leq M$, $|e|\leq E$ and $T\in (0,T^o)$ we have
\begin{align}
&|\bar F^e(x^e,e,T)-\bar F^{Euler}(x^a,e,T)|  \notag \\
&\leq
|F^e(x^e,U(x^e,e,T),T)-F^{Euler}(x^e,U(x^e,e,T),T)|  \notag \\
&+|F^{Euler}(x^e,U(x^e,e,T),T)-F^{Euler}(x^a,U(x^a,e,T),T)|  \notag  \\
&\leq
|F^e(x^e,U(x^e,e,T),T)-F^{Euler}(x^e,U(x^e,e,T),T)| \notag  \\
&+|x^e-x^a|+T|f(x^e,U(x^e,e,T))-f(x^a,U(x^a,e,T))|  \label{eq:a1}\\
&\leq
 (1+KT)|x^e-x^a|+ \bar L T^2 |f(x^e,U(x^e,e,T))| \label{eq:a2} \\
&\leq
 (1+KT)|x^e-x^a|  \notag \\
 &+ \bar L T^2 ( |f(x^e,U(x^e,e,T))-f(0,U(0,e,T))|  +|f(0,U(0,e,T))| ) \notag \\
&\leq (1+KT)|x^e-x^a| + \bar L T^2  \left(K |x^e|+ \phi(|e|)\right) \label{eq:a3}\\
&\leq (1+KT)|x^e-x^a| +  T \rho(T)  (|x^e|+ \phi(|e|)).
\end{align}
In \eqref{eq:a1} we have used the definition of the Euler approximation.
In \eqref{eq:a2} we have used \eqref{eq:desigualdad} from \ref{eq:condicion}) and \eqref{eq:desi1} from Lemma~\ref{lema:21}.
In \eqref{eq:a3} we have used \eqref{eq:can} from \ref{eq:fen0})
and \eqref{eq:desigualdad} from \ref{eq:condicion}).
Note that $\phi\in\Ki$ is given by \ref{eq:fen0}) and hence does not depend on $M$ or $E$. Thus, \eqref{eq:kraftwerk} holds and $(\bar F^{Euler},\bar F^e)$ is REPC.
\mer

\begin{claim}
\label{clm:al}
$(\bar F^{RK},\bar F^{Euler})$ is REPC.
\end{claim}
\indent\emph{Proof of Claim~\ref{clm:al}:}
For the sake of notation, 
define $\bar f(x,e,T)  :=f(x,U(x,e,T))$.
Employing the definitions of the Euler and Runge-Kutta models, we have
\begin{align}
 \big| \bar F^{Euler}(x,e,T)&-\bar F^{RK}(z,e,T)\big| \notag \\
 &\le |x-z| + T \left| \bar f(x,e,T) - \sum_{i=1}^s b_i \bar f(y_i,e,T)\right| \notag
\end{align}
Adding and subtracting $\left(\sum_{i=j}^s b_i\right) \bar f(y_{j-1},e,T)$, for $j=2,\ldots,s$, and operating, we reach
\begin{align*}
    &\big| \bar F^{Euler}(x,e,T)-\bar F^{RK}(z,e,T)\big| \le |x-z|  \\
    &+ T \Bigg| \bar f(x,e,T) \notag \
    - \left[ \sum_{j=2}^s \left(\sum_{i=j}^s b_i\right) 
    \big[\bar f(y_j,e,T) - \bar f(y_{j-1},e,T) \right]  \notag \\
    &+ \left(\sum_{i=1}^s b_i\right) \bar f(y_1,e,T) \big] \Bigg|.
\end{align*}
Taking into account that $\sum_{i=1}^s b_i = 1$ and that $y_1 = z$, then
\begin{align}
    &\big| \bar F^{Euler}(x,e,T)-\bar F^{RK}(z,e,T)\big|
    \le |x-z|   \label{cota}\\
    &+ T \left| \bar f(x,e,T) - \bar f(z,e,T) \right| 
    + TB \sum_{j=2}^s \left| \bar f(y_j,e,T) - \bar f(y_{j-1},e,T) \right|  \notag 
\end{align}
with $B:= \sum_{i=2}^s |b_i|$.
Define $T^{\#}:=\min\{1, T^{u},T^{iv}, T^{v}\}>0$ and 
$A:= 2\max_{i=2,\ldots, s , j=1,\ldots, i-1 } |a_{ij}|$.
Consider $|x|\leq M$, $|e|\leq E$, and $T\in (0,T^{\#})$. From \eqref{charl}
we have
\begin{align}
&|y_{i}|
\leq
|y_1|+  T^{\#} \sum_{j=1}^{i-1} |a_{ij}|\left|\bar f(y_j,e,T)\right| \notag \\
&\leq
|x|+  T^{\#} A\sum_{j=1}^{s-1}\left|f(y_j,U(y_j,e,T)) \right|,\quad\text{for }i=2,\ldots,s.
  \label{este}
\end{align}
Define $M_1:=M$ and recursively for $i=2,\ldots,s$,
\begin{equation*}
  M_i := M_{i-1} + (s-1)T^{\#}AC_f\left(M_{i-1},C_u(M_{i-1},E) \right),
\end{equation*}
where $C_f(\cdot,\cdot)$ and $C_u(\cdot,\cdot)$ are given by Assumptions~\ref{def:bounded} and~\ref{def:controllaw}. With these definitions, it follows that
  $|y_i|\leq M_s \text{ for all } 1\le i \leq s$.
Define 
$T^*:=\min\{1,T^{\#},T^u(M_s,E),T^{iv}(E),T^{v}(M_s,E)\}$.
and $K_s:=\max \{K(M_s,E),1 \}$. 
For $T\in (0,T^*)$, it follows that
\begin{align}
  \left|\bar f(y_1,e,T)\right| 
  &\leq\left|\bar f(y_1,e,T)-\bar f(0,e,T)\right|+\left|\bar f(0,e,T)\right| \notag \\
  &\leq K_s|y_1| +\phi(|e|) \notag
\end{align}
and for $i=2,\ldots,s$
\begin{align}
&\left|\bar f(y_i,e,T)\right|  \notag \displaybreak[0]\\
&\leq|f(y_i,U(y_i,e,T))-f(0,U(0,e,T)|+|f(0,U(0,e,T))| \notag \\
&\leq K_s|y_i| +\phi(|e|) 
\leq K_s\left|y_1+T \sum_{j=1}^{i-1} a_{ij} \bar f(y_j,e,T) \right| +\phi(|e|) \notag\\
&\leq K_s|y_1|+K_s T\left| \sum_{j=1}^{i-1} a_{ij} \bar f(y_j,e,T) \right| +\phi(|e|)\notag \\
&\leq K_s|y_1|+\phi(|e|)+K_sAT  
\sum_{j=1}^{i-1}  \left|\bar f(y_j,e,T) \right|. \label{cage}
\end{align}
Using \eqref{cage} recursively yields
\begin{align}
\left|\bar f(y_i,e,T)\right|  
&\leq 
\left[ K_s |y_1|+ \phi(|e|) \right] \sum_{j=0}^{i-1} (K_sAT)^j
 \notag \\&
 \leq  C
\left( |y_1|+ \phi(|e|) \right),\quad\text{for }i=1,\ldots,s, \label{mozart}
\end{align}
where we have used the fact that $K_s\geq 1$
and defined $C:=K_s\sum_{j=0}^{s-1} (K_sAT^{*})^j$.
From \eqref{cota} and \ref{eq:condicion}), then provided $T\in (0,T^*)$, we have
\begin{align}
    &\big| \bar F^{Euler}(x,e,T)-\bar F^{RK}(z,e,T)\big| \notag \\ 
    &\le (1+KT)|x-z| + TBK_s \sum_{j=2}^s \left|y_j - y_{j-1} \right|. \label{XYZ}
\end{align}
Using \eqref{charl} and defining $a_{ij} := 0$ for $j\ge i$, we have
\begin{align}
|y_i-y_{i+1}| &= \left|  T\sum_{j=1}^{s} (a_{ij}-a_{(i+1)j} ) \bar f(y_j,e,T) \right| \notag \\
 &\leq T A  \sum_{j=1}^{s} |\bar f(y_j,e,T)| 
 \quad \text{for }i=1,\ldots,s-1. \label{brianeno}
\end{align}
Combining \eqref{brianeno} and \eqref{XYZ}, then
\begin{align}
 &|\bar F^{Euler}(x,e,T)-\bar F^{RK}(z,e,T)| \notag\displaybreak[0] \\
        &\leq  (1+KT)|x-z|  + (s-1)ABK_sT^2     \sum_{j=1}^{s} \left|\bar f(y_j,e,T)\right| \label{quan}
\end{align}
Using \eqref{mozart} in \eqref{quan} we obtain,
for all $|x|,|z|\leq M$, $|e|\leq E$ and $T\in(0,T^*)$
\begin{align*}
 &|\bar F^{Euler}(x,e,T)-\bar F^{RK}(z,e,T)|  \notag \\
 &\leq  (1+KT)|x-z| + (s-1) ABK_s T^2 C  \sum_{j=1}^{s}  \left( |y_1|+ \phi(|e|) \right) \\
    &\leq  (1+KT)|x-z| + T\rho(T)   \left( |z|+ \phi(|e|) \right) 
\end{align*}
where $\rho \in \K_\infty$ is defined via $\rho(T):=(s-1)s A BCK_s T$. Note that $\phi\in\Ki$ is given by \ref{eq:fen0}) and hence does not depend on $M$ or $E$. Thus, (\ref{eq:kraftwerk}) holds and $(\bar F^{Euler},\bar F^{RK})$ is REPC.
\mer

According with Claim~\ref{clm:beta} and Claim~\ref{clm:al},
by Proposition~\ref{cor:1} then the pair $(\bar F^{RK},\bar F^{e})$
is REPC. \qed

\section{\textbf{Proof of Lemma~\ref{lema_interbound}}}
\label{app_lema_interbound}

Consider $M,E\geq0$ given and let the S-ISS-VSR property generate the bound \eqref{eq:S-ISS-VSR} and 
$T^\star(M,E)$.
Given some $k \in \N_0$, the evolution of the sampled-data system 
between any consecutive samples $x_k$ and $x_{k+1}$ of the continuous-time solution
that begins from initial condition $|x_0| \leq M$ with $\|\{e_i\}\| \leq E$ and $\{T_i\} \in \Phi(T^\star)$
is given by $x(t_k+t)=\bar F^e(x_k,e_k,t)$, thus
\begin{align}
x(t_k+t)=x_k + \int_{0}^{t} f\left(\bar F^e(x_k,e_k,s),U(x_k,e_k,T_k)\right) ds \label{eq:intersam}
\end{align}
for all $t\in [0,T_k]$.
Define $R:= \beta(M,0)+\gamma(E)$ from
Definition~\ref{def:SE-ISS-VSR}, then $|x_k| \leq R$ for all $k\in \N_0$.
From Assumption \ref{def:controllaw} we have that $|U(x,e,T)|\leq C_u(|x|,|e|)$
for all $|x| \leq R$, $|e|\leq E$ and $T\in (0,T^u(R,E))$.
Define 
$C(|x|,|e|):=C_f(2 |x|,C_u(|x|,|e|))$ with $C_f$
from Assumption~\ref{def:bounded}, 
and $\bar T(M,E):= \min \left\{T^\star,T^u, \frac{R}{C(R,E)} \right\}$.
Next, we will prove that  
$|x(t_k+t)|\leq 2 R$ for all $t\in [0,\bar T]$. 
Let $|x_k|\leq R$, $|e_k|\leq E$
and
define
\begin{align}
    \tau:= \inf \left\{t>0:\left|F^e(x_k,U(x_k,e_k,T_k),t)\right|\geq 2R\right \}. \label{tiempo_}
\end{align}

Due to 
$ F^e(x_k,U(x_k,e_k,T_k),0)=x_k$
and the continuity of 
$ F^e(x_k,U(x_k,e_k,T_k),\cdot)$
we have $\tau>0$ for all $T_k \in [0, \bar T]$.
For a contradiction, suppose that $\tau<\bar T$ for some $T_k \in [0, \bar T]$. From continuity and 
\eqref{tiempo_}, it follows that $\left| F^e(x_k,U(x_k,e_k,T_k), \tau)\right|=2 R$ and
$\left|F^e(x_k,U(x_k,e_k,T_k),t)\right| < 2 R$ for all $t \in [0,\tau)$.
Thus,
\begin{align}
&\left| F^e(x_k,U(x_k,e_k,T_k),\tau)\right| \notag \\
&\leq |x_k| + \int_{0}^{\tau} \left|f\left(F^e(x_k,U(x_k,e_k,T_k),s),U(x_k,e_k,T_k)\right)\right| ds \notag \\
&< R+ C(R,E)  \bar T
\le 2R\label{result1}
\end{align}
The strict inequality in \eqref{result1} 
contradicts $\left| F^e(x_k,U(x_k,e_k,T_k),\tau)\right|=2R$.
Therefore, $\tau \geq \bar T$.
From \eqref{eq:intersam}, \eqref{result1} and \eqref{eq:S-ISS-VSR}, then 
for all $k \in \N_0$, $\{T_i\} \in \Phi(\bar T)$, $|x_0|\leq M$ and $\|\{e_i\}\| \leq E$, we have
\begin{align}
\label{RRRR4}
\left|x(t_k+t)\right| %
&\leq \beta\left(|x_0|,\sum_{i=0}^{k-1}T_i\right)+\gamma\left(\sup_{0\le i \le k-1}|e_i| \right )+ C(R,E)t
\end{align}
for all $t \in [0,T_k]$.

\section*{Acknowledgments}
The authors are grateful to the Associate Editor for the in-depth reading and very constructive suggestions for improvement. Work partially supported by Agencia Nacional de 
Promoción Científica y Tecnológica (ANPCyT), Argentina, under grant
PICT 2018-1385.

\bibliographystyle{elsarticle-harv}
\bibliography{bibliografia201810}

\end{document}